\documentclass[12pt,notitlepage]{article}
\usepackage{hyperref} 
\usepackage[authoryear,round]{natbib}
\usepackage{amssymb,amsmath,multirow,graphicx}
\usepackage{extarrows}
\usepackage{amsthm}
\usepackage{MnSymbol}
\usepackage{indentfirst}
\usepackage{afterpage}
\usepackage{graphicx}
\usepackage{subfigure}
\usepackage{epsfig}
\usepackage{epsf}    
\usepackage{setspace}  
\usepackage{ifthen}

\newcommand{\nothing}[1]{} \nothing{This allows you to make long comments}

\newtheorem{proposition}{Proposition}
\newtheorem{theorem}{Theorem}

\newcommand{\be}{\begin{equation}}
\newcommand{\ee}{\end{equation}}

\begin{document}

\title{Artificial Superintelligence May be Useless: \\ Equilibria in the Economy of Multiple AI Agents}

\author{Huan Cai {\footnote { Department of Economics and Business, Cornell College, Mount Vernon, IA, 52314.}}\,~~
Ziqing Lu {\footnote{Program of Applied Mathematical and Computational Sciences, University of Iowa, Iowa City, IA 52242.}}\,~~
Catherine Xu {\footnote{Iowa City Math Circle and Club, Iowa.}}\,~~
Weiyu Xu {\footnote{Department of Electrical and Computer Engineering, University of Iowa, Iowa City, IA 52242. Email: \texttt{weiyu-xu@uiowa.edu}}}\,~~
Jie Zheng {\footnote{Center for Economic Research, Shandong University, Shandong Province, China.}}
\footnote{The authors are listed alphabetically. }}
\maketitle

\begin{abstract}

With recent development of artificial intelligence, it is more common to adopt AI agents in economic activities. This paper explores the economic actions of agents, including human agents and AI agents, in an economic game of trading products/services, and the equilibria in this economy involving multiple agents.  We derive a range of equilibrium results and their corresponding conditions using a Markov chain stationary distribution based model. One distinct feature of our model is that we consider the long-term utility generated by economic activities instead of their short-term benefits. For the model consisting of two agents, we fully characterize all the possible economic equilibria and conditions. Interestingly, we show that unless each agent can at least double (not merely increase) its marginal utility by purchasing the other agent's products/services, purchasing the other agent's products/services will not happen in any economic equilibrium.  We further extend our results to three and more agents, where we characterize more economic equilibria. We find that in some equilibria, the  ``more powerful'' AI agents contribute zero utility to ``less capable'' agents.   
\end{abstract}



\newpage

\section{Introduction}

When the adoption of advanced AI agents becomes more common, new theories in economics are needed to understand the economic implications of AI agents \citep{hadfield2025economy}. In the new economy, the AI agents may act autonomously, collaborate, and compete with each other and humans. At the current rate of progress in adopting AI agents, it is important for us to investigate the interplays among these AI agents and humans, and to understand the economic equilibria among them.  This will help us best evaluate the benefits and risks of adopting AI agents for humans, from an economic perspective \citep{hammond2025multi}.

There are many open questions in studying the effects of AI agents in economics. AI agents, and AI technologies in general, can greatly increase productivity for given tasks and provide better products and services.  However, from the perspective of multiple humans and AI agents, it is not well understood how these AI agents or AI technologies will affect the economic actions of humans and other AI agents. Nor do we know how diverse AI agents' involvements in economy will contribute to the utilities obtained by various other humans and AI agents.   
 
On the surface, it seems always beneficial for an AI gent or a human to use the help of other powerful AI agents and to adopt AI technologies for boosting productivity or for increasing its own obtained utility. But a fundamental question remains: is it always optimal to adopt AI agents when they are available? As a consumer, a producer, or any player in this new economy of multiple AI agents, how does such an AI agent or human rationally decide whether, when, from whom, and at what price, to purchase the services/products? 

This paper aims to answer these questions by investigating the economic interplays among AI agents and human agents. Different from earlier research on AI agents \citep{goldfarb2019digital,abrardi2022artificial,agrawal2023similarities,kapoor2024ai,li2024more}, instead of investigating the short-term interplays among various agents, we mainly look into the long-term economic interplays and derive their asymptotic (as the number of time epsisodes goes to infinity) utility obtained through a Markov-chain stationary distribution based theoretical model. Using a Markov chain stationary distribution based model for human and AI agents, we derive a range of equilibrium results and the conditions for the equilibria to hold. 

In \citet{cai2023ai}, the authors developed a theory of real price to determine whether any player should adopt certain AI technologies. Using the asymptotic analysis of a Markov chain transition probability matrix modeling the network of AI technologies providers and consumers, this theory shows that the benefits of adopting AI depend on the structure of the AI-producer-consumer network and the utility matrix that measures the benefits provided by AI technologies. In this paper, building on the Markov chain model  provided by this theory, we investigate a multiple-agent economic game using Nash equilibrium analysis. Our model demonstrates that, in an economic equilibrium, the decisions of multiple agents can be quite different from those when they only consider short-term benefits.  One particular interesting result is that if the artificial superintelligence (the ``more powerful'') agents are more powerful than the ``less capable'' AI agents and human agents, in any economic equilibrium, the artificial superintelligence agents will contribute no economic benefits to the ``less capable'' AI agents or human agents. The economic intuitions of these findings provide theoretical guidance for human consumers to make important decisions of adopting AI agents. They also shed light on finding better ways to evaluate the economic contributions from AI agents and make sure that AI agents are well aligned with humans' economic interests.

We first consider the setting of two agents (whether they are AI or humans) and we are able to fully characterize every possible equilibrium of the economic interplays between these two agents. First, we show that it is not always optimal to adopt AI agents (Proposition \ref{equilibrium0} and \ref{equilibrium_optimalwhenp=q=0}) even if the AI agents provide higher utility to buyers. In particular, if an agent cannot at least double its marginal utility by switching from self-production to purchasing products/services of the other agent, such an agent is better off not adopting the products/services of the other agent. More interestingly, even when the ``doubling marginal utility'' conditions are satisfied for both agents, there are three types of  equilibria: a bilateral partial adoption equilibrium where an agent adopts both outside production (provided by the other agent) and  its own self-production (Proposition \ref{equilibrium_optimal}), a bilateral full adoption equilibrium where both agents choose to adopt the products/services of the other agent without any self production (Proposition \ref{equilibrium_optimalwhenp=q=1}), and a unilateral full adoption equilibrium where only one agent chooses to adopt the other agent's products/services without any self-production (Proposition \ref{equilibrium_optimalwhenp=1} and \ref{equilibrium_optimalwhenq=1}). In addition, social welfare is maximized under the full adoption equilibrium. 

Furthermore, we extend our model to interesting equilibria when there are more than two agents. All the results show that the economic benefits brought by the AI agents depend on the producer-consumer network structures and the resulting stationary distribution of currency flow moded by a Markov chain, not only on the myopic single-step benefits of adopting AI agents.  We show that when some ``more powerful'' agents (such as agents with artificial superintelligence) are much more powerful than ``less capable'' agents,  the ``less capable'' agents will not adopt the products/services of the ``more powerful'' agents, in an economic equilibrium, and thus the ``more powerful'' agents contribute zero economic utility to the ``less capable'' agents.

This paper is organized as follows. In Section \ref{Sec:mathmodelMarkov}, we introduce our distinct mathematical model which uses the stationary distribution of the Markov chain probability matrix to evaluate the economic benefits of adopting AI agents, where the Markov chain probability transition matrix describes economic interplays between different agents. In Section \ref{Sec:mathmodel}, we fully characterize the Nash equilibria of two agents' economic interplays. In Sections \ref{sec:stationaryfor3agent} and \ref{sec:3agentequilibrium}, we characterize interesting equilibria for more than 2 agents. In Section \ref{sec: conclusions}, we discuss our results, and conclude this paper by providing summaries and future directions. 

\section{Mathematical Model}\label{Sec:mathmodelMarkov}
We consider $n$ agents, where $n$ is a positive integer. Each agent can function as a producer and also as a consumer. Note that an ``agent'' in our model may represent a human agent, an AI agent, or even an entity of collaborative AI/human agents (for example, the collection of all the people in a country can be modeled as an agent in international trade; the collection of all the AI agents owned and controlled by a technology company can be modeled a single ``big'' agent in providing products/services). 
We consider a spending matrix $P\in \mathbb{R}^{n\times n}$, where $P$ is a Markov chain transition ``probability'' matrix. We denote the element in the $i$-th row and $j$-th column of $P$ as $P_{ij}$. For matrix $P$, we have
$\sum_{i=1}^{n}P_{ij}=1$ for every $1 \leq j \leq n$. The number $P_{ij}$ denotes the fraction of the $j$-th player's currency it spends on purchasing the products/services from the $i$-th agent. Note that the numbers in the Markov chain transition probability matrix do not represent probabilities in their physical meanings: they just denote the fraction of one agent's money spent on purchasing products/services from another agent.  

We use $x^t \in \mathbb{R}^{n\times 1}$ to denote the currency vector at time index $t$, where $t\geq 0$ is an integer.  We assume that each element of $x^t$ is non-negative, representing the currency possessed by each agent at time index $t$. We refer to the time interval between consecutive time indices as one episode.  We  have a utility matrix  $U\in \mathbb{R}^{n\times n}$ whose elements are non-negative. We denote the element in the $i$-th row and the $j$-th column of $U$ as $U_{ij}$, which represents the utility the $i$-th player provides to the $j$-th agent if the $j$-th agent spends one dollar purchasing products/services from agent $i$.   We assume that no agent changes their spending pattern as time evolves, namely a player maintains the relative ratios of spending on each player.

 On the surface, to maximize the utility of the $j$-th agent, the $j$-th agent should spend all its currency on the agent who provides the highest utility for each dollar spent on it, namely the agent $i$ with the largest $U_{ij}$. However, counter-intuitively we argue that this is not true, and the optimal spending practice should depend on the structure of the network of producers and consumers. We notice that at the time index $(t+1)$,  the currency vector is updated to 
$$ x^{t+1}=Px^{t}.$$

\subsection{An agent's optimal strategy maximizing asymptotic long-term utility}
\label{optimization}
We consider the optimal strategy for an agent to maximize its asymptotic long-term utility, namely the utility per episode as time index goes to infinity.  For the purpose of explanation, following \citet{cai2023ai},  without loss of generality, we focus on the first agent and try to determine its optimal spending strategy, namely determining the first column of matrix $P$. Our goal is to maximize the utility obtained by the player $1$ over time.  

We assume that the matrix $P$ is an irreducible Markov chain transition matrix. This condition is often satisfied. For example, a sufficient condition for $P$ to be irreducible is that both of the following conditions hold: (1) for each $j$, $1\leq j \leq n$,  we have the submatrix $P_{/j,/j}$ is irreducible, where $P_{/j,/j}$ is the submatrix of $P$ excluding the $j$-th row and $j$-th column; (2) for each $j$, $1\leq j \leq n$, there exists an agent $k\neq j$ such that $P_{jk}>0$ and there exists an agent $l\neq k$ such that $P_{lj}>0$.

From $x^{t+1}=Px^{t}$, the sum of the elements of the currency vector $x^{t}$'s remains constant as time index $t$ increases. We can scale $x^t$ to $(x')^t$ such that the sum of the elements of $(x')^t$ is equal to 1,  and then view the vector $(x')^t$ as a probability distribution vector (even though the vector represents scaled currency amount rather than probability distribution in its physical meaning). One can then view $P$ as the transition probability matrix of a Markov chain. Due to irreducibility of the Markov chain transition probability matrix, this Markov chain will have a unique stationary distribution. If this Markov chain is further aperiodic, $(x')^t$ will converge to this unique stationary probability distribution vector $x'$ satisfying $Px'=x'$. This means that the currency vector $x^t$ will converge to a unique stationary currency vector $x$ satisfying $Px=x$. Even if the Markov chain is not aperiodic, by the ergodic theorem for irreducible Markov chains, we know the proportion of the time the Markov chain spends on each state converges to that unique stationary distribution as time index increases, with high probability, no matter what the starting distribution $(x')^0$ is.  So when the Markov chain is irreducible, one can always use its unique stationary distribution to calculate the agent's asymptotic utility per episode.

Note that in the stationary distribution, agent $j$ has stationary current amount $x_j$. Because it still has the spending pattern described by $P_{ij}$'s, the average asymptotic utility is $ x_j\sum_{i=1}^{n} P_{ij} U_{ij}$. In order to maximize the utility of the $j$-th agent under the stationary currency vector, we formulate the following optimization problem (\ref{Defn:utilitymaximization}): 
\begin{minipage}[t]{0.48\textwidth}
\begin{align}\label{Defn:utilitymaximization}
&\underset{x,~P_{:j}}{\rm maximize\ }  x_j\sum_{i=1}^{n} P_{ij} U_{ij}\nonumber\\
&{\rm subject\ to\ }~~~~~x\geq 0, \nonumber \\
&~~~~~~~~~~~~~~~~~~~Px=x, \nonumber \\
&~~~~~~~~~~~~~~~~~~~\|x\|_1=\|x^{0}\|_1,\nonumber\\
&~~~~~~~~~~~~~~~~~~~P_{:j}\geq 0,
\end{align}
\end{minipage}
\hfill
\begin{minipage}[t]{0.48\textwidth}
\begin{align}\label{unit_utilitymaximization}
&\underset{x,~P_{:j}}{\rm maximize\ }  x_j\sum_{i=1}^{n} P_{ij} U_{ij}\nonumber\\
&{\rm subject\ to\ }~~~~~x\geq 0, \nonumber \\
&~~~~~~~~~~~~~~~~~~~Px=x, \nonumber \\
&~~~~~~~~~~~~~~~~~~~\|x\|_1=1,\nonumber\\
&~~~~~~~~~~~~~~~~~~~P_{:j}\geq 0.
\end{align}
\end{minipage}
 where $P_{:j}$ is the $j$-th column of $P$, and $\|x\|_1$ represents the $\ell_1$ norm of $x$, namely the sum of the absolute values of elements in $x$. For simplicity, we enforce that the elements of $x$ be non-negative. In this case, the vector $x$ (the stationary currency amount) represents the scaled stationary distribution of the probability transition matrix $P$. Because the system reaches the scaled stationary probability distribution, we must have the second constraint $Px=x$. Because the sum of the elements of $x$ does not change over time, we have $\|x\|_1=\|x^{0}\|_1$. Since $x_j$ is the stationary currency amount agent $j$ will possess in the stationary distribution, we have the objective function in (\ref{Defn:utilitymaximization}).

 Without loss of generality, we take $\|x^{0}\|_1=1$.  Under this assumption, we can change (\ref{Defn:utilitymaximization}) to (\ref{unit_utilitymaximization}). In fact,  (\ref{unit_utilitymaximization}) is a non-linear optimization problem since the objective function involves products of $x_{j}$ and the elements of $P$. To numerically solve (\ref{unit_utilitymaximization}), one can treat  (\ref{unit_utilitymaximization}) as a polynomial optimization problem, and transform it into a sum of squares optimization formulation \citep{parrilo2004sum} leading to a related semidefinite programming problem which can be solved using algorithms for semidefinite programming.

 We further observe that (\ref{unit_utilitymaximization}) becomes a linear programming problem in the other elements of $x$ and the elements of $P$, if we fix a certain value for $x_j$. This observation leads to the following numerical algorithm which avoids using semidefinite programming for polynomial optimization: the semidefinite programming can be time-consuming for large $n$.  In this algorithm, we solve (\ref{unit_utilitymaximization}) by searching over possible values of $x_j$ (for example, grid search). For each examined value of $x_j$, we solve the corresponding linear programming problem. Then among all the values for $x_j$, we select the one which gives the highest objective function value.  Because we only do grid search (or other types of search) over one-dimensional variable $x_j$, and $x_j$ is bounded between $0$ and $1$, this algorithm has relatively low computational complexity.

\section{Equilibria for the two-agent model}\label{Sec:mathmodel}

We first consider a two-agent model where $n=2$. Each agent decides how much money to spend on the other agent's products and services. For every dollar of currency owned by each agent, agent $A$ spends $p$ dollar on agent $B$'s products/services, and agent $B$ spends $q$ dollar on agent $A$'s products/services. The rest of their currency is spent on self-production. That is: agent $A$ spends $1-p$ proportion of its money on itself, and agent $B$ spends $1-q$ proportion of its money on producing/purchasing products/services locally or themselves. An agent's (say, Agent $A$) self-production can also be thought of in the following way:  Agent $A$ purchases products/services from another agent (say, Agent $C$) while Agent $C$ spends back all its currency on purchasing products/services from Agent $A$. 

Thus, we have a spending matrix $P\in \mathbb{R}^{2\times 2}$. It is a Markov chain transition probability matrix:

$$P=
\begin{pmatrix}
    1-p & q\\
    p & 1-q
\end{pmatrix}.$$

We further assume that each agent's spending generates different levels of utility due to different product qualities. When agent $A$ spends currency on itself (which we call as ``self-production'', namely Agent $A$ preforms the production or services itself) , it receives utility of $a$ for each dollar spent; when it spends currency on agent $B$, the  obtained utility is $c$ per spent dollar. Similarly, when agent $B$ spends on itself, it receives per-dollar utility  of $d$; when it spends currency on agent $A$, the per-dollar utility is $b$.  

Thus, we have a utility matrix $Q \in \mathbb{R}^{2\times 2}$ such that each element represents the utility generated from an agent spending one dollar on another agent or itself:

$$U=
\begin{pmatrix}
    a & c\\
    b & d
\end{pmatrix}.$$


Therefore, for every spent dollar, agent $A$ will receive the following amount of utility:
$$u(A)=pb+(1-p)a,$$
and agent $B$ will receive the following amount of utility:
$$u(B)=qc+(1-q)d.$$

\begin{figure}[htb]
\begin{center}
\includegraphics[width=\textwidth]{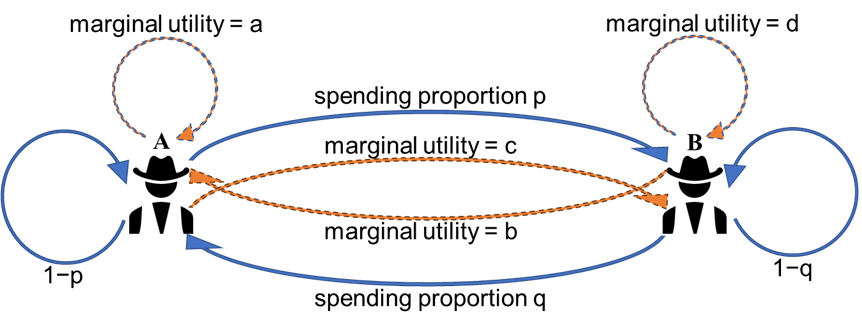}
\caption{An illustration of the two-agent trading game. Agent $A$ spends $p$ proportion of its currency on purchasing products/services from Agent B, and for each dollar spent, Agent $A$ gets marginal utility $b$. It spends $(1-p)$ proportion of its currency on purchasing  products/services from self-production (namely Agent $A$ preforms the production or services itself), and for each dollar spent, Agent $A$ gets marginal utility $a$. Agent $B$ spends $q$ proportion of its currency on purchasing products/services from Agent A, and for each dollar spent, Agent $B$ gets marginal utility $c$. It spends $(1-q)$ proportion of its currency on purchasing  products/services from self-production, and for each dollar spent, Agent $B$ gets marginal utility $d$. An agent's (say, Agent $A$) self-production can also be thought of trading with another dedicated agent:  the agent (Agent $A$) purchases products/services from another hidden agent (say, Agent $C$) which spends back all its currency on purchasing products/services from Agent $A$. }
\label{fig:2playergame}
\end{center}
\end{figure}

\subsection{Stationary Distribution}

For the two-agent example, we use $W^t \in \mathbb{R}^{2\times 1}$ to denote the currency vector at time index $t$, where $t\geq 0$ is an integer.  We assume that each element of $W^t$ is non-negative, representing the possessed currency by each agent at time index $t$. We refer to each time index as one episode.  

We also suppose that  no agent changes their spending pattern as time evolves, namely an agent maintains the relative ratios of spending on each agent.
Without loss of generality, we first use the properties of Markov chain transition matrix to find the stationary distribution (or stationary currency amount). Our goal is to maximize the utility obtained by each agent over time. We assume that at time index $0$, agent $A$ and $B$ has currency $1/2$ respectively (most of our propositions do not need this assumption and work for general initial currency distribution).


In the stationary stage, agent $A$ has currency:
$$W(A)=\frac{q}{p+q},$$
while agent $B$ has currency:
$$W(B)=\frac{p}{p+q}.$$

Therefore, with stationary distribution, agent $A$ has following the utility per episode:
$$V(A)=u(A)W(A)=\frac{q[pb+(1-p)a]}{p+q},$$
and agent $B$ has the following utility:
$$V(B)=u(B)W(B)=\frac{p[qc+(1-q)d]}{p+q}.$$

Next we use these results to find the Nash equilibrium of spending patterns for agents $A$ and $B$ such that none of them can change their spending patterns to improve their asymptotic utility. 

\subsection{Nash Equilibrium}

In the game, each agent decides what actions (spending patterns) to take to maximize the utility given the other agent's action: agent $A$ chooses his spending $p$ given agent $B$'s choice of $q$, whereas agent $B$ chooses her spending $q$ given agent $A$'s choice of $p$. Suppose that the current value for $(p, q)=(p^*, q^*)$. Therefore, we formulate the following optimization problem for agent $A$:
\begin{align}\label{Defn:utilitymaximization_A}
&\underset{p}{\rm max\ }  V(A)=\frac{{q^*}[pb+(1-p)a]}{p+{q^*}}\nonumber\\
&{\rm subject\ to\ }~~~~~0\leq{p}\leq{1}, \nonumber
\end{align}
where ${q^*}$ is given by agent $B$.

Similarly, for agent $B$, the optimization problem is:
\begin{align}
&\underset{q}{\rm max\ }  V(B)=\frac{{p^*}[qc+(1-q)d]}{{p^*}+q}\nonumber\\
&{\rm subject\ to\ }~~~~~0\leq{q}\leq{1}, \nonumber 
\end{align}
where ${p^*}$ is given by agent $A$.

\begin{proposition}\label{equilibrium0}
$p^*=q^*=0$ is always an equilibrium.
\end{proposition}

\begin{proof}
When $p^*=0$, if $q^*>0$, we will have $V(B)=0$; if $q^*=0$, we have $V(B)=d I$ (where $I$ is the amount of currency Agent $B$ has at time index $0$) because agent $B$ spends money only on its self-production and always keeps the money.
When $q^*=0$, if $p^*>0$, we will have $V(A)=0$; when $p^*=0$, we have $V(A)=aJ$ (where $j$ is the amount of currency Agent $B$ has at time index $0$).
Thus $(p^*, q^*)=(0,0)$ is always an equilibrium. 
\end{proof}

\begin{proposition}\label{equilibrium_optimal}
When $b \geq 2a$ and $c \geq 2d$, there exists an equilibrium with ${p^*}=\frac{d}{c-d}$ and ${q^*}=\frac{a}{b-a}$.  Moreover, if $0<p^*<1$ and $0< q^*<1$ happens in an equilibrium, we must have $b>2a$ and $c>2d$. 
\end{proposition}

\begin{proof}




For the given $(p^*, q^*)$ to be an equilibrium, we have the derivative of $V(A)$ at $p^*$ is $0$. Namely for  agent $A$:
\begin{align}
\frac{dV(A)}{dp}=0&\implies\frac{{q^*}[(b-a)(p+{q^*})-pb-(1-p)a]}{(p+{q^*})^2}=0 \nonumber\\
&\implies\frac{{q^*}[(b-a){q^*}-a]}{(p+{q^*})^2}=0\implies{q^*}=\frac{a}{b-a}, \text{or }{q^*}=0. \nonumber
\end{align}

Similarly, for agent $B$:
\begin{align}
\frac{dV(B)}{dq}=0&\implies\frac{{p^*}[(c-d)({p^*}+q)-qc-(1-q)d]}{({p^*}+q)^2}=0 \nonumber\\
&\implies\frac{{p^*}[(c-d){p^*}-d]}{({p^*}+q)^2}=0\implies{p^*}=\frac{d}{c-d}, \text{or }{p^*}=0. \nonumber
\end{align}

From this,  we see that  ${p^*}=\frac{d}{c-d}$ and ${q^*}=\frac{a}{b-a}$ indeed form an equilibrium. 

From the same derivatives being $0$, if $0<p^*<1$ and $0< q^*<1$  happen in an equilibrium, we must have $q^*=\frac{a}{b-a}$ and $p^*=\frac{d}{c-d}$. That implies $b>2a$ and $c>2d$.
\end{proof}


\begin{proposition}\label{equilibrium_optimalwhenp=1}
When $p^*=1$, and $0<q^*<1$ is an equilibrium, then the following two conditions hold true at the same time:
\begin{itemize}
\item $c=2d$;
\item $\frac{b}{a}\geq\frac{1+q^*}{q^*}>2$.
\end{itemize}

Moreover, if $c=2d$ and $\frac{b}{a} > 2$, $p^*=1$ and any $q^*$ such that $\frac{a}{b-a} \leq q^*<1$ is an equilibrium. 

\end{proposition}

\begin{proof}
Because $0<q^*<1$ happens in the equilibrium, for agent $B$:
\begin{align}
\frac{dV(B)}{dq}=0&\implies\frac{{p^*}[(c-d)({p^*}+q)-qc-(1-q)d]}{({p^*}+q)^2}=0 \nonumber\\
&\implies\frac{{p^*}[(c-d){p^*}-d]}{({p^*}+q)^2}=0\implies{p^*}=\frac{d}{c-d}. \nonumber 
\end{align}
Because $p^*=\frac{d}{c-d}=1$, we have $c=2d$.

Moreover, under such a $q^*$, for agent $A$, $V(A)$ is maximized when $p=p^*=1$. This would require that  the derivative of $V(A)$ with respect to $p$ is non-negative, namely
 $$q^*\geq \frac{a}{b-a}.$$

This leads to the proof of these two conditions. 

For the same reason, we can verify that if these two conditions are satisfied,  $p^*=1$ and any $q^*$ such that $\frac{a}{b-a} \leq q^*<1$  maximizes respectively $V(A)$ and $V(B)$ for agents $A$ and $B$. So we have  $p^*=1$ and any $q^*$ such that $\frac{a}{b-a} \leq q^*<1$ is an equilibrium. 

\end{proof}

\begin{proposition}\label{equilibrium_optimalwhenq=1}
When $0<p^*<1$, and $q^*=1$ is an equilibrium, then the following two conditions hold true at the same time:
\begin{itemize}
\item $a=2b$;
\item $\frac{c}{d}\geq\frac{1+p^*}{p^*}>2$.
\end{itemize}

Moreover, if $a=2b$ and $\frac{c}{d} > 2$, $q^*=1$ and any $p^*$ such that $\frac{d}{c-d} \leq p^*<1$ is an equilibrium. 
\end{proposition}

\begin{proof}
By symmetry, the proof of the last proposition works for this case too.
\end{proof}

\begin{proposition}\label{equilibrium_optimalwhenp=q=1}
If $p^*=q^*=1$ is an equilibrium, we must have $b\geq 2a$ and $c\geq 2d$.  Vice versa, when $b\geq 2a$ and $c\geq 2d$, then $p^*=q^*=1$ is an equilibrium.
\end{proposition}

\begin{proof}
Firstly, for agent $A$, $V(A)$ is maximized at $p^*=1$. That would require 
$(b-a)q^*-a\geq 0.$ With $q^*=1$, this is $b\geq 2 a$. By symmetry, for agent $B$, we have 
$c\geq 2d$.

Moreover, when $b\geq 2a$ and $c\geq 2d$, we can check that $p^*=1$ and $q^*=1$ maximize the corresponding utilities for user $A$ and user $B$, proving the second part of the proposition.
\end{proof}

\textbf{Remarks}: When $c>2d$ and $b>2a$, we have three sets of equilibrium points. $(0,0)$, $(1, 1)$ or $(\frac{d}{c-d}, \frac{a}{b-a})$. ~\\

\textbf{Remarks}: When $c=2d$ and $b>2a$, $p^*=1$ and $q^*$ is an equilibrium as long as $1\geq q^*\geq \frac{a}{b-a}$. (besides another equilibrium $(0,0)$)
~\\

\textbf{Remarks}: When $b=2a$ and $c>2d$, $q^*=1$ is an equilibrium as long as $1\geq p^*\geq \frac{d}{c-d}$. (besides another equilibrium $(0,0)$)

\begin{proposition}\label{equilibrium_optimalwhenp=q=0}
Assume that $\max(a,b)>0$ and $\max(c,d)>0$. Then if $b<2a$ or $c<2d$, there is only one equilibrium with $p^*=q^*=0$.
\end{proposition}
\begin{proof}
 Based on the discussions above,  if both $p^*$ and $q^*$ are positive and they form an equilibrium, we must have $b\geq 2a$ and $c\geq 2d$, contradicting the assumption of this proposition. 

Now we show that we cannot have an equilibrium where only one of $p^*$ and $p^*$ is $0$. We prove by contradiction. Without loss of generality, suppose that $p^*=0$ and $p^*>0$. Then in the stationary distribution, Agent $B$ will have stationary mass $0$ and thus get $0$ utility per episode. With the assumption that agent $B$ initially have non-zero ($\frac{1}{2}$) currency, Agent $B$ can strictly improve this utility by spending all its currency on self-production. This shows that $p^*=0$ and $p^*>0$ cannot form an equilibrium. Similarly, we cannot have $p^*>0$ and $q^*=0$.   Therefore, the only equilibrium is $p^*=q^*=0$.
 
\end{proof}

In summary, with two agents, all the potential equilibria are presented in Table \ref{tb:2player_equilibria_summary}.
\begin{table}[!hbt]
\begin{center}
\renewcommand\arraystretch{1.6}
    \begin{tabular}{ | c | l | c | c |}
    \hline
    Scenario & Equilibrium Results & Equilibrium Conditions & Proposition\\ \hline
    1 & No adoption ($p^*=q^*=0$) & always holds & \ref{equilibrium0} and  \ref{equilibrium_optimalwhenp=q=0} \\ \hline
    2 & Bilateral partial adoption & & \\ 
     & (${p^*}=\frac{d}{c-d}$ and ${q^*}=\frac{a}{b-a}$) & $b>2a$ and $c>2d$ & \ref{equilibrium_optimal} \\ \hline
    3 & Bilateral full adoption ($p^*=q^*=1$) & $b\geq2a$ and $c\geq2d$ & \ref{equilibrium_optimalwhenp=q=1} \\ \hline
    4 & Unilateral full adoption & & \\
     & ($p^*=1$, and $0<{q^*}<1$ & $b\leq \frac{q^*+1}{q^*}  a $ and $c=2d$ & \ref{equilibrium_optimalwhenp=1} \\
     & or, $ 0<{p^*}<1$, and $q^*=1$) & or, $b=2a$ and $c\geq \frac{p^*+1}{p^*}d$ & and \ref{equilibrium_optimalwhenq=1} \\ \hline
    \end{tabular}
\end{center}
\caption{A summary of two-agent equilibria}
\label{tb:2player_equilibria_summary}
\end{table}    


\section{Stationary distribution for three agents}
\label{sec:stationaryfor3agent}

Now we consider three-agent network trading products/services. We first derive the stationary distribution for the three-agent network under a given Markoc chain transition matrix $P$.
Note that in this section $P_{ji}$ means that the proportion of money Agent $i$ spends purchasing the services and products of Agent $j$. For every agent $i$, its resource $\sum_{j=1}^3P_{ji}=1$ (Each column sum is $1$). We let $x$ be the $3\times 1$ vector representing the stationary distribution, and thus we have $Px=x$ as follows 

\begin{align}
     \begin{pmatrix}
    P_{11} & P_{12} & P_{13} \\
    P_{21} & P_{22} & P_{23} \\
    P_{31} & P_{32} & P_{33} 
\end{pmatrix} \begin{pmatrix}
    x_1 \\ x_2 \\ x_3
\end{pmatrix}=\begin{pmatrix}
    x_1 \\ x_2 \\ x_3 \nonumber
\end{pmatrix}
\end{align}
This is equivalent to
\begin{align}
    (-P_{21}-P_{31})x_1 + P_{12}x_2 + P_{13}x_3 &= 0, \\
    P_{21}x_1 + (-P_{12}-P_{32})x_2 + P_{23}x_3 &= 0, \\
    P_{31}x_1 + P_{32}x_2 + (-P_{13}-P_{23})x_3 &= 0 , \\
    x_1 + x_2 + x_3 &= 1. \label{stationary_system}
\end{align}

We solve for the stationary expenditure $(x_1, x_2,x_3)$ with parameters $P_{ji}'s$. We first rewrite $x_3=1-x_1-x_2$, substitute $x_3$ into the first equation of (\ref{stationary_system}), and get
\begin{align*}
    &(-P_{21}-P_{31})x_1+P_{12}x_2+P_{13}(1-x_1-x_2)=0 \\
    \implies &(-P_{21}-P_{31}-P_{13})x_1+(P_{12}-P_{13})x_2+P_{13}=0.
\end{align*}
We also substitute $x_3$ into the second equation in (\ref{stationary_system}) and get
\begin{align*}
    &P_{21}x_1+(-P_{12}-P_{32})x_2+P_{23}(1-x_1-x_2)=0 \\
   \implies &(P_{21}-P_{23})x_1+(-P_{12}-P_{32}-P_{23})x_2+P_{23}=0.
\end{align*}
Then we solve for the simplified system with two variables $x_1, x_2$,
\begin{align}
\begin{cases}
    (-P_{21}-P_{31}-P_{13})x_1+(P_{12}-P_{13})x_2+P_{13}=0, \\
    (P_{21}-P_{23})x_1+(-P_{12}-P_{32}-P_{23})x_2+P_{23}=0.
\end{cases}\label{simplified-system}
\end{align}
Assume the determinant $\Delta = (P_{12}+P_{32}+P_{23})(P_{21}+P_{31}+P_{13})-(P_{12}-P_{13})(P_{21}-P_{23}) \neq 0$. To solve for $x_1$, we multiply $(-P_{12}-P_{32}-P_{23})$ to the first equation, and $(P_{12}-P_{13})$ to the second equation of the simplified system (\ref{simplified-system}). We get
\begin{equation*}
    \resizebox{\textwidth}{!}{$
    x_1
    =\Big(P_{13}(P_{12}+P_{32}+P_{23})+P_{23}(P_{12}-P_{13})\Big)\Big/\Big((P_{12}+P_{32}+P_{23})(P_{21}+P_{31}+P_{13})-(P_{12}-P_{13})(P_{21}-P_{23})\Big)
    $}.
\end{equation*}
Similarly, to solve for $x_2$, we multiply $(P_{21}-P_{23})$ to the first equation, and $(-P_{21}-P_{31}-P_{23})$ to the second equation of system (\ref{simplified-system}). We get
\begin{equation*}
    \resizebox{\textwidth}{!}{$
    x_2
    =\Big(P_{23}(P_{21}+P_{31}+P_{13})+P_{13}(P_{21}-P_{23})\Big)\Big/\Big((P_{12}+P_{32}+P_{23})(P_{21}+P_{31}+P_{13})-(P_{12}-P_{13})(P_{21}-P_{23})\Big) $}.
\end{equation*}
Therefore, the solution to the stationary expenditure is
\begin{align}
\begin{pmatrix}
    x_1\\x_2\\x_3
\end{pmatrix}=\begin{pmatrix}
    \Big(P_{13}(P_{12}+P_{32}+P_{23})+P_{23}(P_{12}-P_{13})\Big)/\Delta \\
    \Big(P_{23}(P_{21}+P_{31}+P_{13})+P_{13}(P_{21}-P_{23})\Big)/\Delta\\
    1-x_1-x_2
\end{pmatrix}.
\end{align}


\section{Equilibria for three and more agents}
\label{sec:3agentequilibrium}
In Section \ref{Sec:mathmodel}, we have found all the equilibria for the case of two agents. Now we turn to the case of three and more agents, building on the results in Section  \ref{sec:stationaryfor3agent}. 
We have the following theorems. 
\begin{theorem}
Let $n$ be the number of agents and $n\geq 3$. Then there always exists at least one equilibrium, where each agent does not spend on any other agent's products/services. 
\end{theorem}
\begin{proof}
We prove that each agent choosing to purchase only its own products/services is an equilibrium. In fact, if a agent (say, Agent $1$) chooses any other spending pattern where this agent spends a non-zero proportion of its currency on purchasing another agent's services/products, then the stationary distribution for Agent $1$ will be $0$ because no other agent is  purchasing from Agent $1$.  This will make the utility obtained by Agent $1$ zero, not possibly increasing its utility. Thus each agent only spends its currency on purchasing its own products/services is an equilibrium. 

\end{proof}

\begin{theorem}
\label{thm:agentisolated}
Suppose that Agent $1$ is an agent who provides a larger utility for itself than any other agent provides for Agent $1$, per spent dollar spent by Agent $1$.  Then the following equilibrium exists:
a) Agent $1$ only purchases from itself; b) Agent $2$ and Agent $3$ do not purchase from Agent $1$; c) Agent $2$ and Agent $3$ achieve the equilibrium described in the two-agent case involving only Agent $2$ and Agent $3$. 

Moreover, in each possible equilibrium, Agent $1$ must be spending its currency purchasing its own products/services; other agents must spend $0$ currency purchasing Agent $1$'s products/services if Agent $2$ and Agent $3$ alone can reach an equilibrium where each of them has achieved a non-zero utility per episode.
\end{theorem}

\begin{proof}
We start with proving the first part. Suppose that Agent $1$ instead chooses to spend a non-zero proportion of its currency on purchasing other agents' products and services, and other agents do not change their spending patterns. Then in the stationary distribution, the currency flowing to Agent $1$ and the currency flowing from Agent $1$ to other agents must be equal. But that currency flow must be $0$ since other agents do not spend on purchasing products/services from Agent $1$. This in turn means that the money flow from Agent $1$ to other agents must be equal to $0$. So the stationary probability of Agent $1$ must be $0$, thus meaning Agent $1$ will only have $0$ utility. Thus Agent $1$ cannot change its current spending pattern to improve its utility. 

For the same reason,  Agent $2$ or $3$ cannot improve their utility by spending its currency on purchasing Agent $1$'s products and services.  Thus if Agent $2$ and Agent $3$ are already in a two-agent equilibrium involving only themselves,  the three-agent system is in an equilibrium.

Now we prove the second part. Suppose in an equilibrium, Agent $1$ spends a non-zero fraction of its currency on Agent $2$ or Agent $3$'s products/services. Then Agent $1$ can switch that fraction to buying its own products/services. By doing so, Agent $1$ will not reduce its stationary distribution probability, since there will be no currency flow going from Agent $1$ to other agents (Note that previously, the net currency flow is $0$ at Agent $1$).  In this process, the utility for Agent $1$ is increased since Agent $1$  itself provides higher products/services utility per dollar spent.  Suppose that Agent $2$ or Agent $3$ spends a non-zero fraction on Agent $1$'s products and services, then Agent $2$'s (or Agent $3$'s) stationary probability must be equal to $0$, leading to a zero utility for Agent $2$ (Agent $3$). Thus all the other agents will not spend on purchasing products/services from Agent $1$, if they can reach an equilibrium where each of them has a non-zero utility per episode.

\end{proof}

\begin{theorem}
Suppose that each agent provides unit utility for any other agent who spends one dollar purchasing its products/services, but each agent provides $0$ utility if it purchases its own products/services.  Then the following actions for each agent form an equilibrium: 
(1) Each agent spends one half of its currency on each of the other two agents; (2) Each agent spends $0$ on purchasing its own services/products. 
\end{theorem}
\begin{proof}
Due to symmetry, we only consider Agent $1$ and see whether it can change its action to increase its utility while assuming the other two agents keep their actions fixed.  Namely, $$P_{12}=P_{32}=P_{13}=P_{23}=1/2,$$ and 
$$P_{22}=P_{33}=0.$$
$$P_{22}=P_{33}=0.$$
So by the derivations in Section \ref{sec:stationaryfor3agent}, the stationary probability (mass) for Agent $1$ is

$$x_1=\frac{\frac{1}{2}\times (\frac{1}{2}+\frac{1}{2}+\frac{1}{2})}{\frac{3}{2}(P_{21}+P_{31}+\frac{1}{2})}=\frac{1}{2P_{21}+2P_{31}+1}$$.

Under this stationary mass, the utility for Agent $1$ is given by 

$$ x_1 \times (P_{21}\times 1+P_{31} \times 1)=  \frac{P_{21}+P_{31}}{2P_{21}+2P_{31}+1}.$$

This quantity is maximized when $P_{21}+P_{31}$ is maximized. Thus $P_{21}=P_{31}=\frac{1}{2}$ maximizes Agent $1$'s utility, proving that the set of actions considered achieve the Nash equilibrium. 

\end{proof}

Interestingly, we show that the following ``rotationally''-symmetric strategy is not a Nash equilibrium, even though it achieves the same utility for Agent 1 as the discussed strategy in the last theorem.

\begin{theorem}
Suppose that each agent provides unit utility for any other agent who spends one dollar purchasing its products/services, but provides $0$ utility if it purchases its own products/services.  Then the following actions for each agent do not form an equilibrium: 
(1) Agent $1$ spends all its currency on Agent $2$; (2) Agent $2$ spends all its currency on Agent $3$; (3) Agent $3$ spends all its currency on Agent $1$. 
\end{theorem}

\begin{proof}
One can show that the stationary mass is $$x_{1}=x_{2}=x_{3}=\frac{1}{3}.$$
The utility for Agent $1$ is $\frac{1}{3}$.
Let us consider Agent $1$, and let Agent $1$ switch its spending strategy from ``spending all the currency on Agent $2$'' to ``spending all the currency on Agent $3$''.

After the switch, there will be a new stationary distribution which can be verified as 

$$x_1=x_3=\frac{1}{2},~ x_{2}=0.$$
The new utility for Agent $1$ is $\frac{1}{2}$. Thus the earlier strategy is not a Nash equilibrium. 

\end{proof}

From the discussions above, we can see that, even though in both strategies the agents choose the agents which provide the best myopic utilities for them, they may differ in terms of whether being a Nash equilibrium.

In the next theorem, we provide an equilibrium when a player controls two agents or when two agents collaborate. 

\begin{theorem}
Suppose that Agent $1$ provides utility $c$ to Agent $2$ or Agent $3$ at their spending of one dollar.  Suppose that Agent $2$ or $3$ provides utility $b$ to Agent $1$ at Agent $1$'s spending of one dollar. 
Suppose that Agent $2$ provides utility $d$ to itself if it spends one dollar on purchasing its own products/services. Suppose that Agent $3$ provides utility $d$ to itself if it spends one dollar on purchasing its own products/services. Suppose that Agent $1$ provides utility $a$ to itself if it spends one dollar purchasing its own products/services.  We suppose that Agent $2$ ($3$) provides $0$ utility to Agent $3$ ($2$) if Agent $3$ spends one dollar on purchasing from Agent $2$ ($3$).

Suppose that Agent $2$ and Agent $3$ collaborate to maximize their combined utility or a player controls both Agents $2$ and $3$ to maximize their combined utility. We assume that Agent $2$ and Agent $3$ do not spend currency purchasing each other's products/services; and Agent $2$ and Agent $3$ spend the same proportion of their currency on purchasing from Agent $1$.

Then the following actions form an equilibrium: 
(1) Agent $1$ spends $\frac{d}{2(c-d)}$ of its currency on purchasing products/services from Agent $2$; Agent $1$ spends $\frac{d}{2(c-d)}$ of its currency on purchasing products/services from Agent $3$; Agent $1$ spends the remaining proportion on purchasing its own products/services; (2) Agent $2$ spends $\frac{a}{b-a}$ fraction on purchasing products/services from Agent $1$; Agent $2$ spends the remaining fraction on purchasing products/services from itself; (3)Agent $3$ spends $\frac{a}{b-a}$ fraction on purchasing products/services from Agent $1$; Agent $3$ spends the remaining fraction on purchasing products/services from itself. 

\end{theorem}

\begin{proof}

We first focus on Agent $1$. From its perspective, Agent $2$ and Agent $3$ work as a single ``combined'' agent which spends $\frac{a}{b-a}$ fraction of their currency purchasing products/services from Agent $1$, and the remaining fraction on purchasing products/services from the ``combined'' agent itself. Thus, from Proposition \ref{equilibrium_optimal}, Agent $1$ cannot improve its utility by changing its current spending actions.

Moreover, due to the assumptions that Agent $2$ and Agent $3$ spend the same proportion of their money on purchasing products/services from Agent $1$; and they do not spend currency on each other,  they cannot improve their utility by changing their actions. 

\end{proof}

\begin{theorem}
Suppose that there are four agents. Two of them are ``more powerful'' AI agents, and the other two are ``less capable'' humans.  We assume that, for any buying agent,  any selling  agent from the group of ``more powerful'' agents will provide a strictly bigger utility than any selling agent from the group of ``less capable'' agents. 

Then the following is an equilibrium: 1) Agent $1$ and Agent $2$ purchase products/services from each other or from themselves; 2) Agent $3$ and Agent $4$ purchase products/services from each other or from themselves. 3) There are no purchases between any other agent pair.

 Moreover, any possible equilibrium must satisfy: 1) No ``less capable'' agent spends non-zero currency in purchasing anything from any ``more powerful'' agent; 2) No ``more powerful'' agent spends non-zero currency in purchasing anything from any ``less capable'' agent.

\end{theorem}

\begin{proof}

We start with proving the first part of this theorem. 
First of all, any ``more powerful'' agent will not purchase from any ``less capable'' agent. This is due to the same reason as explained in the proof of Theorem \ref{thm:agentisolated}:  A ``more powerful'' Agent $1$ (or $2$) purchasing form less capable agents while the ``less capable'' agents do not buy from any of the ``more powerful'' agents will result in the stationary probability or mass for Agent $1$  getting to $0$. This will in turn lead to a $0$ utility for the ``more powerful'' agent, not increasing the utility for Agent $1$. 

Secondly, any ``less capable'' agent (Agent 3, say) will not purchase from any of the ``more powerful'' agents. This is because of the same reason: due to ``more powerful'' agents not purchasing from any ``less capable'' agents, Agent $3$ purchasing from any ``more powerful'' agent will result in its stationary mass being set to $0$, not improving its utility. 

Now we prove the second part of this theorem.  We first argue that any ``more powerful'' agent will not purchase from the ``less capable'' agent. Suppose instead one ``more powerful'' agent (say, Agent $1$) purchases products/services from the ``less capable'' agents. Then Agent $1$ can switch all such spending on the ``less capable'' agents to spending on itself. In this switch, both Agent $1$'s stationary mass and its utility per dollar will increase, thus leading to a higher asymptotic utility.  This proves the second claim. 

Since no ``more powerful'' agent spends on any ``less capable'' agent, no ``less capable'' agent will spend on any ``more powerful'' agent because doing so will lead to its stationary mass being $0$. 

\end{proof}

In the equilibria above, one may wonder why the optimal strategy of an agent is not to purchase from the most powerful supplier agent who provides the highest utility per dollar spent on it.  This is because in our model, the spending pattern or strategy of an agent not only affects its instant utility obtained from spending, but also affects the stationary distribution.  In the following, we give an example where the optimal strategy for an agent is not to go for the most powerful supplier.  

\subsection{An example where an agent does not go for the most powerful supplier in its optimal strategy}
We consider a three-agent network for simplicity (even though the method and codes of calculating the optimal strategies work for general $n$). In this network, we focus on optimizing the spending strategy for Agent $1$, given that Agents $2$'s and $3$'s spending strategies are fixed.  We consider the following Markov chain matrix:
$$
P=\begin{pmatrix}
    P_{11} & P_{12} & P_{13}\\
    P_{21} & P_{22} & P_{23}\\
    P_{31} & P_{32} & P_{33}
\end{pmatrix},  
$$
where we need to optimize $P_{11}$, $P_{21}$ and $P_{31}$. Here $P_{ij}$ means the proportion of Agent $j$'s currency spent on purchasing Agent $i$'s products/services.  In our experiment, we have the following numerical values for the last two columns of $P$:
$$
P=\begin{pmatrix}
    P_{11} & \alpha & 0.5\\
    P_{21} & 1-\alpha-0.02 & 0.01\\
    P_{31} & 0.02 & 0.49
\end{pmatrix},  
$$
where $\alpha$ is a tunable parameter for Agent $2$. The larger $\alpha$ is, more willing is Agent $2$ to spend their currency on purchasing Agent $1$'s products and services. 
In this example, we take $\alpha=0.01$.

We take the following $U$ matrix

$$
U=\begin{pmatrix}
    0 & 1 & 1\\
    9.8 & 0 & 1\\
    1 & 1 & 2
\end{pmatrix}.
$$

Even though Agent $2$ provides a much higher utility (9.8) for each dollar spent by Agent $1$, the optimized strategy of Agent $1$ is to spend all its currency on Agent $3$. This is because this strategy will give a bigger stationary distribution to Agent $1$. We note that Agent $3$ spends a much larger proportion of its currency on purchasing from Agent $1$ than Agent $2$ does.

\label{numerical}


\section{Conclusions and Future Directions}
\label{sec: conclusions}
In this paper, we consider the equilibria in economy of multiple agents, including human agents and AI agents. We derived the equilibria based on a Markov chain stationary distribution based model of trading products/services, in which we consider the long-term utility generated by economic activities instead of their short-term benefits. For the model consisting of two agents, we fully characterize all the possible economic equilibria and the corresponding conditions for them to happen. Interestingly, we show that unless each agent can at least double (not merely increase) its marginal utility by purchasing the other agent's products/services, purchasing the other agent's products/services will not happen in any economic equilibrium.  We further extend our results to three and more agents, where we characterize more economic equilibria. We find that in some equilibria, the  ``more powerful'' AI agents contribute $0$ utility to ``less capable'' agents.  

Our results show that the benefits of adopting AI agents depend  not only  on their capabilities/productivities but also on the graph structure of the producer-consumer network consisting of human and AI agents as nodes.   

In future work, we will extend the results to consider specific and different types of products/services provided by each agent. Besides asymptotic utility per episode, we will also consider the discounted utility which balances between long-term and short-term utility.  It is also interesting to investigate the scenario where the total amount of currency increases due to central bank policies and increase in overall productivity.  In this paper, we only consider the scenario where each agent acts autonomously, and in future works, we can consider the scenario where certain players can control a set of agents or where several agents can collaborate and have utilities as functions of AI agents' joint actions and joint spendings. We would like to remark that our method of using the Markov chain probability transition matrix to model economic activities can potentially be applied to other research topics such as international trade and investments, tariff policies, economic and financial sanctions in interconnected economies, and buy-local economic policies.  We note that in these topics, economic activities can be described by an economic interaction network where the involved entities are network nodes.


\clearpage
\begin{spacing}{1.6}
\typeout{[Bibliography]}
\bibliography{Reference_20260227,Reference_NullSpaceCondition}
\end{spacing}
\end{document}